\documentclass[11pt]{article}

\usepackage[T1]{fontenc}
\usepackage[utf8]{inputenc}
\usepackage[margin=1in]{geometry}
\usepackage{amsmath,amssymb,amsthm,mathtools,bm}
\usepackage{enumitem}
\usepackage{natbib}
\usepackage{microtype}
\usepackage[colorlinks=true,allcolors=blue]{hyperref}

\numberwithin{equation}{section}

\theoremstyle{plain}
\newtheorem{theorem}{Theorem}[section]
\newtheorem{proposition}[theorem]{Proposition}
\newtheorem{lemma}[theorem]{Lemma}
\newtheorem{corollary}[theorem]{Corollary}

\theoremstyle{definition}
\newtheorem{definition}[theorem]{Definition}
\newtheorem{example}[theorem]{Example}

\theoremstyle{remark}
\newtheorem{remark}[theorem]{Remark}

\newcommand{\R}{\mathbb{R}}
\newcommand{\1}{\mathbf{1}}
\newcommand{\T}{\mathcal{T}}
\newcommand{\Xdep}{\mathcal{X}}
\newcommand{\diag}{\operatorname{diag}}
\newcommand{\supp}{\operatorname{supp}}

\newcommand{\Tr}{\operatorname{Tr}}
\newcommand{\norm}[1]{\left\lVert #1 \right\rVert}
\newcommand{\inner}[2]{\left\langle #1, #2 \right\rangle}

\title{A Unified Theory of Ownership Concentration, Overlap, and Dependence\\[0.4em]
\large Margins, Cells, Spectral Structure, and Linear Transmission}
\author{Miquel Noguer i Alonso \and Iro Tasitsiomi}
\date{\today}

\begin{document}
\maketitle

\begin{abstract}
Ownership concentration is not a scalar. For a normalized investor--stock matrix $A$, it has three irreducible layers: concentration across investors, concentration across stocks, and dependence in the joint assignment of investors to stocks. This paper develops a unified quadratic framework for those layers and shows that the same residual operator that measures static overlap also governs linearized market transmission. Raw micro concentration $M(A)=\sum_{i,j}A_{ij}^2$ admits exact row and column decompositions, support bounds, and fixed-marginal extremal characterizations on the transportation polytope. Benchmark-adjusted dependence
\[
\Xdep(A)=\sum_{i,j}\frac{(A_{ij}-p_i s_j)^2}{p_i s_j}
\]
admits two exact decompositions: it is a size-weighted average of investor-level deviations from the market portfolio and, symmetrically, of stock-level deviations from the investor base. The paper also proves a multiscale aggregation law: under any partition of investors, total dependence splits exactly into between-group dependence and within-group heterogeneity. Spectrally, $\Xdep(A)$ equals the sum of squared nontrivial singular values of the whitened matrix $D_p^{-1/2} A D_s^{-1/2}$. The residual operator $L$ then yields two dynamic consequences: idiosyncratic fire-sale vulnerability is bounded by the dominant overlap mode $\rho(A)$, while aggregate benchmark-relative alpha variance has worst-case capacity $\rho(A)^2$ and isotropic average-case capacity $\Xdep(A)$. The fixed-marginal geometry also motivates a feasible-range sparsity score that benchmarks observed micro concentration against the sharp minimum and maximum implied by the marginals. The resulting framework separates scale concentration, feasible sparsity, overlap, and linear transmission in a way that is mathematically transparent and empirically usable for work on crowding, fragility, and systemic risk.
\end{abstract}

\section{Introduction}

Diversification has been central to modern portfolio theory since \citet{markowitz1952}. At the same time, Herfindahl--Hirschman concentration measures remain the standard way to summarize how economic mass is distributed across firms, sectors, or agents \citep{herfindahl1950,hirschman1945}. In ownership data, however, concentration is inherently two-sided: there is a distribution across investors and a distribution across stocks. A market can be concentrated on one margin and diffuse on the other. Once that is recognized, the natural primitive is no longer a vector, but a nonnegative matrix.

That shift matters economically as well as mathematically. Research on stock fragility shows that concentrated ownership and common ownership shocks are distinct channels of vulnerability \citep{greenwoodthesmar2011}. Related work on overlapping portfolios emphasizes that common asset holdings create indirect network linkages that can amplify stress \citep{greenwoodlandierthesmar2015,gualdi2016,poledna2021}. Recent evidence from bond mutual funds points in the same direction: ownership concentration is associated with higher asset-price volatility \citep{huangwangwang2024}. Yet many empirical summaries still collapse ownership to one side of the system, losing the joint structure of \emph{who owns what}.

This paper develops a compact mathematical theory of that joint structure. The central claim is that ownership concentration has three distinct layers:
\begin{enumerate}[leftmargin=2em]
    \item \emph{investor concentration}, captured by the Herfindahl index of row marginals;
    \item \emph{stock concentration}, captured by the Herfindahl index of column marginals;
    \item \emph{ownership dependence}, captured by departures of the full matrix from the proportional benchmark $ps^\top$.
\end{enumerate}
Raw cell concentration $M(A)=\sum_{i,j}A_{ij}^2$ is informative, but by itself it mixes marginal scale and assignment structure. The paper separates those components explicitly.

The contribution of this study sits at the intersection of several established findings. Fixed-margin ownership sets are transportation polytopes and contingency-table spaces \citep{kleewitzgall1968,diaconisgangolli1995,deloerakim2014}. The normalized spectral representation is the correspondence-analysis decomposition of a contingency table \citep{greenacre1984,greenacre2007}. Those ingredients are classical in their own domains. The contribution here is to combine them into an ownership calculus, and to push that calculus one step further through new aggregation, comparative-static, and transmission results natural to the ownership setting.

The main contributions are as follows.
\begin{enumerate}[leftmargin=2em]
    \item The paper defines a three-layer measurement architecture for ownership systems: investor concentration, stock concentration, and benchmark-adjusted dependence.
    \item It proves exact decompositions of raw micro concentration $M(A)$ into size-weighted within-investor and within-stock concentration measures.
    \item It proves exact decompositions of benchmark-adjusted dependence $\Xdep(A)$ into investor-level and stock-level $\chi^2$ deviations from the market benchmark.
    \item It proves a multiscale aggregation theorem: any investor partition yields an exact between-group/within-group decomposition of $\Xdep(A)$, with investor mergers as a sharp corollary.
    \item It studies the fixed-marginal geometry of ownership matrices as a transportation polytope, characterizing the minimizers and maximizers of $M(A)$ and introducing a feasible-range sparsity score.
    \item It shows that the entire nontrivial dependence structure is encoded in the singular spectrum of the whitened ownership matrix.
    \item It proves a dynamic bridge from static overlap to linear transmission: $\rho(A)$ is the sharp worst-case coefficient for idiosyncratic fire-sale propagation, while $\Xdep(A)$ is the isotropic average-case capacity for benchmark-relative active variance.
\end{enumerate}

Three features are especially important for originality. First, the paper does not simply import correspondence analysis into finance; it identifies the ownership meaning of the nontrivial singular modes and places them alongside separate marginal concentration measures. Second, the aggregation law for $\Xdep(A)$ reveals a genuinely multiscale structure: coarse ownership data retain the between-group component exactly while discarding a precisely quantified within-group heterogeneity term. Third, the same residual operator $L$ governs linear transmission in two dual directions: idiosyncratic investor liquidation shocks pass through $L^\top$, while centered stock-return dispersion passes through $L$. This turns a static overlap index into a dynamic capacity theory.

The broader lesson is simple but important. Ownership concentration is not a single number. It is a layered property of a bipartite system, and any serious analysis of crowding or fragility should keep those layers distinct.

\section{Setup: the ownership matrix and its marginals}\label{sec:setup}

Consider $n$ investors and $m$ stocks. Let
\[
Y=(Y_{ij})\in\R_+^{n\times m}
\]
denote the raw holdings matrix, where $Y_{ij}\ge 0$ is the dollar amount or exposure of investor $i$ to stock $j$. Let total represented wealth be
\[
T=\sum_{i=1}^n\sum_{j=1}^m Y_{ij}>0,
\]
and define the normalized matrix
\[
A_{ij}=\frac{Y_{ij}}{T},
\qquad
\sum_{i=1}^n\sum_{j=1}^m A_{ij}=1.
\]
Thus $A$ is a probability matrix on investor--stock pairs.

Define the row and column marginals by
\[
p_i=\sum_{j=1}^m A_{ij},
\qquad
s_j=\sum_{i=1}^n A_{ij}.
\]
Then $p$ is the investor-size distribution and $s$ is the stock-size distribution.

\begin{definition}[Marginal concentration]
The \emph{investor concentration} and \emph{stock concentration} are
\[
H_I=\sum_{i=1}^n p_i^2,
\qquad
H_S=\sum_{j=1}^m s_j^2.
\]
Their inverses
\[
N_I=\frac{1}{H_I},
\qquad
N_S=\frac{1}{H_S}
\]
are the corresponding \emph{effective numbers} of investors and stocks.
\end{definition}

If wealth is equally split across $q$ investors, then $H_I=1/q$ and $N_I=q$. If capitalization is equally split across $r$ stocks, then $H_S=1/r$ and $N_S=r$.

\begin{definition}[Raw micro concentration]
The \emph{raw micro concentration} of the ownership matrix is
\[
M(A)=\sum_{i=1}^n\sum_{j=1}^m A_{ij}^2=\norm{A}_F^2=\Tr(A^\top A).
\]
Its inverse
\[
N_M=\frac{1}{M(A)}
\]
is the quadratic effective number of occupied investor--stock cells.
\end{definition}

\begin{remark}[Collision probabilities]
If two ownership dollars are drawn independently from $A$, then
\[
H_I=\mathbb{P}(\text{same investor}),
\qquad
H_S=\mathbb{P}(\text{same stock}),
\qquad
M(A)=\mathbb{P}(\text{same investor and same stock}).
\]
Thus $H_I$ and $H_S$ are marginal collision probabilities, while $M(A)$ is the joint collision probability.
\end{remark}

\section{Raw micro concentration: exact decompositions and support bounds}\label{sec:micro}

For every investor with $p_i>0$, define the within-investor portfolio weights
\[
q_{ij}=\frac{A_{ij}}{p_i},
\qquad
\sum_{j=1}^m q_{ij}=1,
\]
and let
\[
c_i=\sum_{j=1}^m q_{ij}^2
\]
denote investor $i$'s within-portfolio concentration. For every stock with $s_j>0$, define the within-stock owner shares
\[
r_{ij}=\frac{A_{ij}}{s_j},
\qquad
\sum_{i=1}^n r_{ij}=1,
\]
and let
\[
d_j=\sum_{i=1}^n r_{ij}^2
\]
denote stock $j$'s within-owner concentration.

\begin{proposition}[Two exact decompositions]\label{prop:microdecomp}
The raw micro concentration satisfies
\[
M(A)=\sum_{i:p_i>0} p_i^2 c_i
      =\sum_{j:s_j>0} s_j^2 d_j.
\]
\end{proposition}

\begin{proof}
By substitution,
\[
\sum_{i:p_i>0} p_i^2 c_i
=\sum_{i:p_i>0} p_i^2\sum_{j=1}^m\left(\frac{A_{ij}}{p_i}\right)^2
=\sum_{i,j}A_{ij}^2
=M(A).
\]
The stock-side identity follows analogously.
\end{proof}

Proposition \ref{prop:microdecomp} is the first key identity of the paper. It shows that raw cell concentration is both a size-weighted average of investor specialization and a size-weighted average of stock-owner concentration.

Let
\[
k_i=\#\{j:A_{ij}>0\},
\qquad
\ell_j=\#\{i:A_{ij}>0\}
\]
denote the row and column support sizes.

\begin{lemma}[Local support bounds]
For every investor with $p_i>0$,
\[
\frac{1}{k_i}\le c_i\le 1,
\]
and for every stock with $s_j>0$,
\[
\frac{1}{\ell_j}\le d_j\le 1.
\]
\end{lemma}

\begin{proof}
The vector $(q_{ij})_{j:A_{ij}>0}$ is a probability vector of length $k_i$. Its squared norm is minimized by the uniform allocation and maximized by a point mass. The same reasoning applies to $(r_{ij})_{i:A_{ij}>0}$.
\end{proof}

\begin{theorem}[Global bounds]\label{thm:globalmicro}
The raw micro concentration satisfies
\[
\sum_{i:p_i>0}\frac{p_i^2}{k_i}\le M(A)\le H_I,
\qquad
\sum_{j:s_j>0}\frac{s_j^2}{\ell_j}\le M(A)\le H_S.
\]
Consequently,
\[
\max\!\left\{\frac{H_I}{m},\frac{H_S}{n}\right\}\le M(A)\le \min\{H_I,H_S\}.
\]
\end{theorem}

\begin{proof}
Apply Proposition \ref{prop:microdecomp} together with the local bounds on $c_i$ and $d_j$. The coarse bounds follow from $k_i\le m$ and $\ell_j\le n$.
\end{proof}

\begin{corollary}[Uniform support caps]
If each investor holds at most $K$ stocks and each stock is held by at most $L$ investors, then
\[
\frac{H_I}{K}\le M(A)\le H_I,
\qquad
\frac{H_S}{L}\le M(A)\le H_S.
\]
Equivalently,
\[
N_I\le N_M\le K N_I,
\qquad
N_S\le N_M\le L N_S.
\]
\end{corollary}

The point is already clear at this stage. The raw matrix quantity $M(A)$ is not redundant relative to $H_I$ and $H_S$, but it is not a pure dependence measure either. It blends marginal size and assignment structure. The next section isolates the dependence component.

\section{Benchmark-adjusted ownership dependence}\label{sec:dependence}

A natural benchmark is \emph{proportional ownership}: every investor holds the market portfolio and differs only in scale. In matrix form,
\[
A=ps^\top,
\qquad
A_{ij}=p_i s_j.
\]
Under this benchmark,
\[
M(ps^\top)=\left(\sum_i p_i^2\right)\left(\sum_j s_j^2\right)=H_I H_S.
\]
Example \ref{ex:notmin} below shows, however, that $ps^\top$ is not generally the fixed-marginal minimizer of $M(A)$. It is a no-dependence benchmark rather than a universal minimum-concentration benchmark.

\begin{definition}[Ownership dependence index]
Assume from now on that the active rows and columns satisfy $p_i>0$ and $s_j>0$. Define
\[
\Xdep(A)=\sum_{i=1}^n\sum_{j=1}^m \frac{(A_{ij}-p_i s_j)^2}{p_i s_j}.
\]
\end{definition}

\begin{remark}
The index $\Xdep(A)$ is the Pearson $\chi^2$-divergence between the joint law $A$ and the product law $p\otimes s$ \citep{pearson1900}. It vanishes exactly when every investor is a scaled copy of the market portfolio.
\end{remark}

\begin{definition}[Investor and stock profile divergences]
For any probability vectors $u$ and $v$ on the same active support, define
\[
\chi^2(u\|v)=\sum_k \frac{(u_k-v_k)^2}{v_k}.
\]
For investor $i$, the row profile is $q_i=(q_{i1},\dots,q_{im})$, and for stock $j$, the owner profile is $r_j=(r_{1j},\dots,r_{nj})$.
\end{definition}

\begin{proposition}[Equivalent forms and profile decompositions]\label{prop:profiledecomp}
The dependence index satisfies
\[
\Xdep(A)=\sum_{i,j}\frac{A_{ij}^2}{p_i s_j}-1.
\]
Equivalently, if
\[
R_{ij}=\frac{A_{ij}}{p_i s_j},
\]
then
\[
\Xdep(A)=\sum_{i,j}p_i s_j (R_{ij}-1)^2.
\]
Most importantly,
\[
\Xdep(A)=\sum_{i=1}^n p_i\,\chi^2(q_i\|s)
       =\sum_{j=1}^m s_j\,\chi^2(r_j\|p).
\]
\end{proposition}

\begin{proof}
Expanding the square yields
\[
\Xdep(A)
=\sum_{i,j}\frac{A_{ij}^2}{p_i s_j}
 -2\sum_{i,j}A_{ij}
 +\sum_{i,j}p_i s_j
=\sum_{i,j}\frac{A_{ij}^2}{p_i s_j}-1.
\]
The likelihood-ratio form follows immediately from $A_{ij}=p_i s_j R_{ij}$.

For the row decomposition, write $A_{ij}=p_i q_{ij}$:
\[
\Xdep(A)
=\sum_{i,j}\frac{(p_i q_{ij}-p_i s_j)^2}{p_i s_j}
=\sum_i p_i \sum_j \frac{(q_{ij}-s_j)^2}{s_j}
=\sum_i p_i \chi^2(q_i\|s).
\]
The column formula follows analogously.
\end{proof}

Proposition \ref{prop:profiledecomp} is the dependence analogue of Proposition \ref{prop:microdecomp}. It shows that $\Xdep(A)$ is a size-weighted average of investor deviations from the market portfolio and also a size-weighted average of stock deviations from the market investor base.

\begin{remark}[Benchmark-adjusted, not margin-free]
The index $\Xdep(A)$ adjusts for the proportional benchmark $ps^\top$, but it is not literally margin-free. The weights $1/(p_i s_j)$ are part of the geometry. A better description is therefore \emph{benchmark-adjusted dependence}: $\Xdep(A)$ removes the rank-one market-weight component while respecting the observed marginals.
\end{remark}

\subsection{A multiscale aggregation law}

The dependence index admits an exact hierarchical decomposition under investor aggregation.

\begin{theorem}[Multiscale aggregation of dependence]\label{thm:multiscale}
Let $\Pi=\{G_1,\dots,G_g\}$ be a partition of the investor set $\{1,\dots,n\}$. For each group $G_a$, define
\[
P_a=\sum_{i\in G_a} p_i,
\qquad
\bar q_a=\frac{1}{P_a}\sum_{i\in G_a} p_i q_i.
\]
Then
\[
\Xdep(A)
=
\underbrace{\sum_{a=1}^g P_a\,\chi^2(\bar q_a\|s)}_{\text{between-group dependence}}
+
\underbrace{\sum_{a=1}^g \sum_{i\in G_a} p_i \sum_{j=1}^m \frac{(q_{ij}-\bar q_{aj})^2}{s_j}}_{\text{within-group heterogeneity}}.
\]
In particular, investor aggregation can only decrease benchmark-adjusted dependence.
\end{theorem}

\begin{proof}
For vectors over stocks, define the weighted inner product
\[
\inner{x}{y}_{s^{-1}}:=\sum_{j=1}^m \frac{x_j y_j}{s_j},
\qquad
\norm{x}_{s^{-1}}^2:=\inner{x}{x}_{s^{-1}}.
\]
By Proposition \ref{prop:profiledecomp},
\[
\Xdep(A)=\sum_{a=1}^g \sum_{i\in G_a} p_i \norm{q_i-s}_{s^{-1}}^2.
\]
For each $i\in G_a$, decompose
\[
q_i-s=(q_i-\bar q_a)+(\bar q_a-s).
\]
Hence
\[
\norm{q_i-s}_{s^{-1}}^2
=
\norm{q_i-\bar q_a}_{s^{-1}}^2
+2\inner{q_i-\bar q_a}{\bar q_a-s}_{s^{-1}}
+\norm{\bar q_a-s}_{s^{-1}}^2.
\]
After weighting by $p_i$ and summing over $i\in G_a$, the cross term vanishes because
\[
\sum_{i\in G_a} p_i (q_i-\bar q_a)=0.
\]
Therefore
\[
\sum_{i\in G_a} p_i \norm{q_i-s}_{s^{-1}}^2
=
\sum_{i\in G_a} p_i \norm{q_i-\bar q_a}_{s^{-1}}^2
+
P_a \norm{\bar q_a-s}_{s^{-1}}^2.
\]
Summing over $a$ yields the claimed identity.
\end{proof}

Theorem \ref{thm:multiscale} is the paper's central structural result on dependence. It shows that benchmark-adjusted ownership dependence is an exact multiscale object. Coarser ownership data do not merely \emph{approximate} fine-level dependence: they retain the between-group component while discarding a precisely quantified within-group heterogeneity term.

\begin{corollary}[Exact merger law for dependence]\label{cor:mergerX}
Suppose investors $a$ and $b$ merge into a single entity. Let $\tilde A$ denote the merged ownership matrix, and let
\[
\tilde q=\frac{p_a q_a+p_b q_b}{p_a+p_b}
\]
be the merged portfolio profile. Then
\[
\Xdep(A)-\Xdep(\tilde A)
=
\frac{p_a p_b}{p_a+p_b}\sum_{j=1}^m \frac{(q_{aj}-q_{bj})^2}{s_j}\ge 0.
\]
Equality holds if and only if the two investors have identical portfolio weights: $q_a=q_b$.
\end{corollary}

\begin{proof}
Apply Theorem \ref{thm:multiscale} to the partition whose only nontrivial block is $\{a,b\}$. The within-block term becomes
\[
p_a \norm{q_a-\tilde q}_{s^{-1}}^2+p_b \norm{q_b-\tilde q}_{s^{-1}}^2
=
\frac{p_a p_b}{p_a+p_b}\norm{q_a-q_b}_{s^{-1}}^2.
\]
This is exactly the dependence lost under the merger.
\end{proof}

Corollary \ref{cor:mergerX} already shows that raw concentration and benchmark-adjusted dependence respond differently to consolidation. Merging investors averages row profiles and thereby weakly reduces dependence, even though it weakly increases the raw concentration of ownership. Section \ref{sec:operations} makes this contrast explicit.

\section{Fixed marginals and transportation-polytope geometry}\label{sec:transport}

Fix probability vectors $p\in\R_+^n$ and $s\in\R_+^m$ with the same total mass. The fixed-marginal ownership set is
\[
\T(p,s)=\{A\in\R_+^{n\times m}:A\1=p,\ A^\top \1=s\}.
\]
This is the classical transportation polytope associated with $(p,s)$ \citep{kleewitzgall1968,diaconisgangolli1995,brualdi2006,deloerakim2014}.

\begin{theorem}[Optimization of raw micro concentration on fixed marginals]\label{thm:transportopt}
For fixed marginals $(p,s)$:
\begin{enumerate}[label={\rm(\roman*)},leftmargin=2em]
    \item the problem $\min_{A\in\T(p,s)} M(A)$ has a unique solution;
    \item the problem $\max_{A\in\T(p,s)} M(A)$ admits a solution at an extreme point of $\T(p,s)$.
\end{enumerate}
\end{theorem}

\begin{proof}
The set $\T(p,s)$ is compact and convex. The map $A\mapsto M(A)=\norm{A}_F^2$ is continuous and strictly convex, so the minimum exists and is unique. Since $M$ is convex, a maximizer can be chosen at an extreme point: if a maximizer were a nontrivial convex combination of two feasible matrices, then at least one endpoint would also be a maximizer.
\end{proof}

\begin{remark}[KKT characterization of the minimizer]
The unique minimizer is the Euclidean projection of the zero matrix onto $\T(p,s)$. The KKT conditions imply the existence of multipliers $(\lambda_i)$ and $(\mu_j)$ such that
\[
A_{ij}^\star=\max\!\left\{0,\frac{\lambda_i+\mu_j}{2}\right\},
\]
with the row and column constraints determining the multipliers. Thus the minimum-concentration solution has an additive-threshold structure.
\end{remark}

To characterize the maximizers more sharply, consider the support graph $G(A)$: a bipartite graph with investor vertices $\{1,\dots,n\}$, stock vertices $\{1,\dots,m\}$, and an edge $(i,j)$ whenever $A_{ij}>0$.

\begin{theorem}[Extreme points are exactly the acyclic supports]\label{thm:forest}
Assume the active rows and columns satisfy $p_i>0$ and $s_j>0$. Then $A\in\T(p,s)$ is an extreme point if and only if its support graph $G(A)$ is a forest. Consequently,
\[
\#\supp(A)\le n+m-1
\]
for every extreme point. Therefore, some maximizer of $M(A)$ over $\T(p,s)$ has at most $n+m-1$ positive cells.
\end{theorem}

\begin{proof}
If $G(A)$ contains a cycle, assign alternating signs $\pm 1$ around the cycle and let $B$ be the matrix supported on those edges with the corresponding signs. Every row and column sum of $B$ is zero. For sufficiently small $\varepsilon>0$, both $A+\varepsilon B$ and $A-\varepsilon B$ remain nonnegative and belong to $\T(p,s)$, and
\[
A=\frac12(A+\varepsilon B)+\frac12(A-\varepsilon B)
\]
is a nontrivial convex decomposition. Hence $A$ is not extreme.

Conversely, suppose $G(A)$ is a forest. Let $A=\lambda B+(1-\lambda)C$ with $0<\lambda<1$ and $B,C\in\T(p,s)$. Set $E=B-C$. Then every row and column sum of $E$ is zero. Moreover, if $A_{ij}=0$, nonnegativity of $B$ and $C$ forces $B_{ij}=C_{ij}=0$, so $\supp(E)\subseteq \supp(A)$. Because $G(A)$ is a forest, it has a leaf vertex. The unique edge incident to that leaf must carry zero $E$-mass because the corresponding row or column sum of $E$ is zero. Remove that edge and continue inductively through the forest. All entries of $E$ vanish, so $B=C=A$. Thus $A$ is extreme.

The support-size bound follows because a forest on $n+m$ vertices has at most $n+m-1$ edges.
\end{proof}

\begin{definition}[Feasible-range sparsity score]
For fixed marginals $(p,s)$ with $M_{\min}(p,s)<M_{\max}(p,s)$, define
\[
\Psi(A\mid p,s)
=
\frac{M(A)-M_{\min}(p,s)}{M_{\max}(p,s)-M_{\min}(p,s)},
\qquad A\in\T(p,s).
\]
\end{definition}

\begin{proposition}
For every $A\in\T(p,s)$,
\[
0\le \Psi(A\mid p,s)\le 1.
\]
Moreover, $\Psi(A\mid p,s)=0$ if and only if $A$ is the unique fixed-marginal minimizer of $M$, whereas $\Psi(A\mid p,s)=1$ if and only if $A$ is a fixed-marginal maximizer of $M$.
\end{proposition}

\begin{proof}
This follows immediately from the definition and Theorem \ref{thm:transportopt}.
\end{proof}

The score $\Psi(A\mid p,s)$ answers a different question from $\Xdep(A)$. The quantity $\Xdep(A)$ asks \emph{how far is the ownership matrix from proportional ownership?} In contrast, $\Psi(A\mid p,s)$ asks \emph{how sparse or concentrated is the ownership matrix relative to the sharp feasible range allowed by the marginals?} The two objects are complementary, not interchangeable.

\begin{proposition}[The $2\times2$ transport family]\label{prop:2x2}
Let $p=(a,1-a)$ and $s=(b,1-b)$ with $a,b\in(0,1)$. Every matrix in $\T(p,s)$ has the form
\[
A(x)=
\begin{pmatrix}
x & a-x\\[4pt]
b-x & 1-a-b+x
\end{pmatrix},
\qquad
x\in I_{a,b}:=\bigl[\max\{0,a+b-1\},\,\min\{a,b\}\bigr].
\]
Moreover,
\[
M(A(x))=x^2+(a-x)^2+(b-x)^2+(1-a-b+x)^2,
\]
so the unconstrained minimizer is
\[
x^\star=\frac{a+b}{2}-\frac14,
\]
and the constrained minimizer is obtained by projecting $x^\star$ onto $I_{a,b}$.
\end{proposition}

\begin{proof}
The row and column constraints imply the stated form. Substituting into $M(A)$ yields the quadratic expression. Differentiation gives
\[
\frac{d}{dx}M(A(x))=8x-4a-4b+2,
\]
whose unique zero is $x^\star=(a+b)/2-1/4$. Since the second derivative equals $8$, the constrained minimum is obtained by projection onto the feasible interval.
\end{proof}

\begin{example}[Proportional ownership is not generally the fixed-marginal minimum]\label{ex:notmin}
Under the proportional benchmark, the $2\times2$ family in Proposition \ref{prop:2x2} has $x=ab$. This equals the fixed-marginal minimizer only when
\[
ab=\frac{a+b}{2}-\frac14
\qquad\Longleftrightarrow\qquad
(2a-1)(2b-1)=0.
\]
Thus proportional ownership is not generally the fixed-marginal minimum of $M(A)$.

For instance, if $a=b=0.9$, then
\[
A_{\mathrm{prod}}=
\begin{pmatrix}
0.81 & 0.09\\
0.09 & 0.01
\end{pmatrix},
\qquad
M(A_{\mathrm{prod}})=0.6724,
\]
whereas the true minimizer over $\T(p,s)$ is
\[
A_{\min}=
\begin{pmatrix}
0.80 & 0.10\\
0.10 & 0
\end{pmatrix},
\qquad
M(A_{\min})=0.66.
\]
The benchmark $ps^\top$ is therefore the natural \emph{no-dependence} reference point, but not a universal minimum of raw micro concentration.
\end{example}

\section{Spectral structure of ownership dependence}\label{sec:spectral}

Assume from now on that active rows and columns satisfy $p_i>0$ and $s_j>0$. Define the diagonal matrices
\[
D_p=\diag(p_1,\dots,p_n),
\qquad
D_s=\diag(s_1,\dots,s_m),
\]
and the whitened ownership matrix
\[
K=D_p^{-1/2} A D_s^{-1/2}.
\]
Let
\[
u=(\sqrt{p_1},\dots,\sqrt{p_n})^\top,
\qquad
v=(\sqrt{s_1},\dots,\sqrt{s_m})^\top.
\]
Then $\norm{u}_2=\norm{v}_2=1$.

\begin{remark}[Connection to correspondence analysis]
The centered matrix
\[
L:=D_p^{-1/2}(A-ps^\top)D_s^{-1/2}=K-u v^\top
\]
is exactly the standardized residual matrix of correspondence analysis \citep{greenacre1984,greenacre2007}. The difference here is interpretive: in the present setting, the nontrivial singular modes of $K$ represent ownership-overlap modes beyond the rank-one market benchmark.
\end{remark}

\begin{theorem}[Contractive normalization]\label{thm:contractive}
The matrix $K$ satisfies
\[
Kv=u,
\qquad
K^\top u=v.
\]
Moreover,
\[
\norm{K}_{op}\le 1,
\]
so its largest singular value is exactly $1$.
\end{theorem}

\begin{proof}
Using the row and column constraints,
\[
Kv=D_p^{-1/2} A D_s^{-1/2}v
   =D_p^{-1/2}A\1
   =D_p^{-1/2}p
   =u,
\]
and similarly $K^\top u=v$.

For the norm bound, let $y\in\R^m$ and set $z_j=y_j/\sqrt{s_j}$. Then
\[
\norm{Ky}_2^2
=
\sum_{i=1}^n \frac{1}{p_i}\left(\sum_{j=1}^m A_{ij} z_j\right)^2
=
\sum_{i=1}^n p_i\left(\sum_{j=1}^m q_{ij} z_j\right)^2.
\]
By Jensen's inequality,
\[
\left(\sum_j q_{ij} z_j\right)^2\le \sum_j q_{ij} z_j^2.
\]
Therefore
\[
\norm{Ky}_2^2
\le
\sum_i p_i \sum_j q_{ij} z_j^2
=
\sum_j s_j z_j^2
=
\sum_j y_j^2
=
\norm{y}_2^2.
\]
So $\norm{K}_{op}\le 1$. Because $Kv=u$ and both vectors have unit norm, equality is attained at the top singular value.
\end{proof}

\begin{theorem}[Spectral decomposition of benchmark-adjusted dependence]\label{thm:spectral}
Let
\[
1=\sigma_1(K)\ge \sigma_2(K)\ge\cdots\ge \sigma_r(K)>0
\]
be the positive singular values of $K$. Then
\[
\Xdep(A)=\norm{K-u v^\top}_F^2=\sum_{k=2}^{r}\sigma_k(K)^2.
\]
\end{theorem}

\begin{proof}
By definition,
\[
\norm{K-u v^\top}_F^2
=
\sum_{i,j}\left(\frac{A_{ij}}{\sqrt{p_i s_j}}-\sqrt{p_i s_j}\right)^2
=
\sum_{i,j}\frac{(A_{ij}-p_i s_j)^2}{p_i s_j}
=
\Xdep(A).
\]
Set $L=K-u v^\top$. By Theorem \ref{thm:contractive},
\[
Lv=0,
\qquad
L^\top u=0.
\]
Hence $K$ decomposes orthogonally into the rank-one component $u v^\top$ and the residual $L$, whose singular values are precisely $\sigma_2(K),\dots,\sigma_r(K)$. Therefore
\[
\Xdep(A)=\norm{L}_F^2=\sum_{k=2}^{r}\sigma_k(K)^2.
\]
\end{proof}

\begin{corollary}[Dominant overlap mode]
Define
\[
\rho(A):=\sigma_2(K).
\]
Then
\[
\rho(A)=0
\qquad\Longleftrightarrow\qquad
A=ps^\top.
\]
Moreover,
\[
\rho(A)^2\le \Xdep(A)\le (\min\{n,m\}-1)\rho(A)^2.
\]
\end{corollary}

\begin{proof}
The first statement follows because all nontrivial singular values vanish if and only if $K=u v^\top$, i.e.\ $A=ps^\top$. The bounds follow from Theorem \ref{thm:spectral} and the ordering $\sigma_2(K)\ge \sigma_3(K)\ge\cdots$.
\end{proof}

The spectral picture is therefore complete. The first singular value of $K$ is fixed at one by the marginals. Everything nontrivial in ownership dependence lies in the remaining singular spectrum.

\section{Dynamic implications of the spectral residual}\label{sec:dynamic}

The whitened residual
\[
L=D_p^{-1/2}(A-ps^\top)D_s^{-1/2}=K-u v^\top
\]
is not merely a descriptive statistic. It is the operator that transmits nontrivial ownership structure into linearized market dynamics. The trivial rank-one mode $u v^\top$ carries market-wide motion; the residual $L$ carries the orthogonal, crowding-driven component. This section makes that statement exact in two directions: liquidation shocks propagate through $L^\top$, while cross-sectional stock-return dispersion propagates through $L$. These results give formal content to the fragility intuition of \citet{greenwoodthesmar2011} and the overlapping-portfolio perspective of \citet{poledna2021}.

For vectors on the investor and stock sides, write
\[
\norm{x}_{D_p}^2:=x^\top D_p x,
\qquad
\norm{z}_{D_s}^2:=z^\top D_s z.
\]

\subsection{Spectral bound on fire-sale vulnerability}

Let $\delta\in\R^n$ denote the vector of investor-level liquidation rates, so that investor $i$ sells the dollar amount $\delta_i A_{ij}$ of stock $j$. Define the $D_p$-orthogonal decomposition
\[
\delta_\parallel = \left(\sum_{i=1}^n p_i \delta_i\right)\1,
\qquad
\delta_\perp = \delta-\delta_\parallel.
\]
Then $\delta_\parallel$ is the market-wide liquidation component, and $\delta_\perp$ is the idiosyncratic component with weighted mean zero. The aggregate sell pressure across stocks is
\[
F=A^\top \delta.
\]
Under a linear-impact normalization in which relative price impact equals sell pressure divided by stock size, the impact vector is
\[
\Delta P = D_s^{-1}F = D_s^{-1}A^\top\delta.
\]
We measure severity by the capitalization-weighted quadratic price impact $\norm{\Delta P}_{D_s}^2$.

\begin{theorem}[Spectral bound on fire-sale vulnerability]\label{thm:firesale}
For every liquidation shock $\delta\in\R^n$,
\[
\norm{\Delta P}_{D_s}^2
=
\norm{\delta_\parallel}_{D_p}^2
+
\norm{L^\top D_p^{1/2}\delta_\perp}_2^2
\le
\norm{\delta_\parallel}_{D_p}^2
+
\rho(A)^2\norm{\delta_\perp}_{D_p}^2.
\]
In particular, every purely idiosyncratic shock satisfies
\[
\norm{\Delta P}_{D_s}\le \rho(A)\norm{\delta}_{D_p}
\qquad\text{whenever}\qquad
\sum_{i=1}^n p_i\delta_i=0.
\]
\end{theorem}

\begin{proof}
Set
\[
y=D_p^{1/2}\delta\in\R^n.
\]
Write $y=\alpha u+y_\perp$, where $\alpha=u^\top y$ and $y_\perp\perp u$. Since $D_p^{1/2}\1=u$, one has
\[
\alpha u = D_p^{1/2}\delta_\parallel,
\qquad
 y_\perp = D_p^{1/2}\delta_\perp.
\]
Also,
\[
\Delta P = D_s^{-1}A^\top\delta = D_s^{-1/2}K^\top y.
\]
Therefore
\[
\norm{\Delta P}_{D_s}^2 = \norm{K^\top y}_2^2.
\]
Using $K^\top u=v$ and $L^\top u=0$, we obtain
\[
K^\top y = K^\top(\alpha u+y_\perp)=\alpha v+L^\top y_\perp.
\]
The two terms are orthogonal because
\[
v^\top L^\top y_\perp = (Lv)^\top y_\perp = 0.
\]
Hence
\[
\norm{\Delta P}_{D_s}^2 = \alpha^2 + \norm{L^\top y_\perp}_2^2.
\]
Now $\alpha^2=\norm{\delta_\parallel}_{D_p}^2$ and $\norm{y_\perp}_2^2=\norm{\delta_\perp}_{D_p}^2$. Since $\norm{L}_{op}=\rho(A)$ by Theorem \ref{thm:spectral},
\[
\norm{L^\top y_\perp}_2\le \rho(A)\norm{y_\perp}_2.
\]
Substituting yields the claimed identity and bound.
\end{proof}

The theorem identifies $\rho(A)$ as the sharp idiosyncratic transmission coefficient of the ownership network in the linearized fire-sale system. When $A=ps^\top$, one has $L=0$, and every centered liquidation shock is annihilated after size normalization: only the market-wide component remains.

\subsection{Benchmark-relative active-variance capacity}

Let $R\in\R^m$ be a vector of realized excess returns, centered relative to market capitalization:
\[
s^\top R = \sum_{j=1}^m s_j R_j = 0.
\]
Investor $i$'s benchmark-relative active return is
\[
\alpha_i = \sum_{j=1}^m \left(\frac{A_{ij}}{p_i}-s_j\right)R_j.
\]
Equivalently,
\[
\alpha = D_p^{-1}(A-ps^\top)R.
\]
Define the aggregate active variance by
\[
V(\alpha)=\norm{\alpha}_{D_p}^2=\sum_{i=1}^n p_i\alpha_i^2.
\]
Let
\[
\tilde R=D_s^{1/2}R.
\]
Then $\tilde R\perp v$ because $s^\top R=0$.

\begin{theorem}[Active-variance capacity]\label{thm:alpha}
For every centered return vector $R\in\R^m$,
\[
V(\alpha)=\norm{L\tilde R}_2^2
\le
\rho(A)^2\sum_{j=1}^m s_j R_j^2.
\]
More generally, if $\tilde R$ is a random centered return vector with covariance matrix
\[
\Sigma=\mathbb E[\tilde R\tilde R^\top]
\qquad\text{satisfying}\qquad
\Sigma v=0,
\]
then
\[
\mathbb E[V(\alpha)] = \Tr(L\Sigma L^\top).
\]
In particular, under isotropic centered return dispersion,
\[
\Sigma = \sigma^2(I-vv^\top),
\]
one has
\[
\mathbb E[V(\alpha)] = \sigma^2 \Xdep(A).
\]
\end{theorem}

\begin{proof}
Since
\[
L\tilde R = D_p^{-1/2}(A-ps^\top)D_s^{-1/2}D_s^{1/2}R
= D_p^{-1/2}(A-ps^\top)R,
\]
one gets
\[
\alpha = D_p^{-1/2}L\tilde R.
\]
Therefore
\[
V(\alpha)=\norm{D_p^{1/2}\alpha}_2^2=\norm{L\tilde R}_2^2.
\]
Because $\tilde R\perp v$ and the singular values of $L$ are $\sigma_2(K),\sigma_3(K),\dots$, one has $\norm{L}_{op}=\rho(A)$, so
\[
V(\alpha)\le \rho(A)^2\norm{\tilde R}_2^2=
\rho(A)^2\sum_{j=1}^m s_j R_j^2.
\]
For the covariance formula,
\[
\mathbb E[V(\alpha)]
=
\mathbb E[\tilde R^\top L^\top L\tilde R]
=
\Tr(L^\top L\Sigma)
=
\Tr(L\Sigma L^\top).
\]
If $\Sigma=\sigma^2(I-vv^\top)$, then $Lv=0$ implies
\[
L\Sigma L^\top = \sigma^2 L(I-vv^\top)L^\top = \sigma^2 LL^\top.
\]
Taking traces and using Theorem \ref{thm:spectral} yields
\[
\mathbb E[V(\alpha)] = \sigma^2\norm{L}_F^2 = \sigma^2\Xdep(A).
\]
\end{proof}

Theorem \ref{thm:alpha} gives a worst-case/average-case duality. The dominant mode $\rho(A)$ controls the maximum benchmark-relative dispersion generated by any centered return realization of a given weighted size, whereas $\Xdep(A)$ is the total quadratic capacity of the ownership system to generate benchmark-relative active variance under isotropic centered return dispersion.

\begin{remark}[Scope of the linearization]
The fire-sale and active-variance results are first-order statements. They do not require that real markets obey permanent linear impact or isotropic centered return geometry in full generality. The point is structural: once one linearizes around the proportional benchmark, the same residual operator $L$ necessarily governs the nontrivial transmission component.
\end{remark}

\section{Structural operations and comparative statics}\label{sec:operations}

This section records how the indices behave under canonical ownership operations.

\subsection{Investor mergers}

Suppose investors $a$ and $b$ merge into a single entity, and let $\tilde A$ be the new matrix obtained by summing rows $a$ and $b$.

\begin{proposition}[Mergers raise raw concentration but reduce dependence]\label{prop:merger}
Under an investor merger:
\begin{enumerate}[label={\rm(\roman*)},leftmargin=2em]
    \item investor concentration rises:
    \[
    \tilde H_I-H_I=2p_a p_b\ge 0;
    \]
    \item stock concentration is unchanged:
    \[
    \tilde H_S=H_S;
    \]
    \item raw micro concentration rises:
    \[
    M(\tilde A)-M(A)=2\sum_{j=1}^m A_{aj}A_{bj}\ge 0;
    \]
    \item benchmark-adjusted dependence falls:
    \[
    \Xdep(A)-\Xdep(\tilde A)
    =
    \frac{p_a p_b}{p_a+p_b}\sum_{j=1}^m \frac{(q_{aj}-q_{bj})^2}{s_j}\ge 0.
    \]
\end{enumerate}
Equality in (iii) holds if and only if the two investors hold disjoint sets of stocks, and equality in (iv) holds if and only if the two investors have identical portfolio weights.
\end{proposition}

\begin{proof}
The identities for $\tilde H_I$ and $\tilde H_S$ are immediate from the new marginals. For $M$,
\[
M(\tilde A)-M(A)
=
\sum_{j=1}^m\bigl[(A_{aj}+A_{bj})^2-A_{aj}^2-A_{bj}^2\bigr]
=
2\sum_{j=1}^m A_{aj}A_{bj}.
\]
Part (iv) follows from Corollary \ref{cor:mergerX}.
\end{proof}

Proposition \ref{prop:merger} cleanly separates two notions that are often conflated. Consolidation increases raw concentration because mass is being combined, but it decreases benchmark-adjusted dependence because profile heterogeneity is being averaged out. The wedge between $M$ and $\Xdep$ is therefore economically meaningful.

\subsection{Stock removal}

Suppose stock $j_0$ is removed and the remaining matrix is renormalized. Let
\[
w=1-s_{j_0},
\qquad
\hat A_{ij}=\frac{A_{ij}}{w}
\quad (j\ne j_0).
\]

\begin{proposition}[Stock removal and renormalization]\label{prop:removal}
After removing stock $j_0$ and renormalizing:
\begin{enumerate}[label={\rm(\roman*)},leftmargin=2em]
    \item the new investor marginals are
    \[
    \hat p_i=\frac{p_i-A_{ij_0}}{w},
    \]
    with investors of zero remaining wealth dropped from the active set;
    \item the new raw micro concentration is
    \[
    M(\hat A)=\frac{1}{w^2}\left(M(A)-\sum_{i=1}^n A_{ij_0}^2\right);
    \]
    \item writing $\varepsilon=s_{j_0}$,
    \[
    M(\hat A)=M(A)+2\varepsilon\, M(A)+O(\varepsilon^2)
    \qquad (\varepsilon\to 0).
    \]
\end{enumerate}
\end{proposition}

\begin{proof}
Only the last column is removed, and every remaining cell is divided by $w$. This yields (i) and (ii). For (iii), note that
\[
\sum_i A_{ij_0}^2\le \left(\sum_i A_{ij_0}\right)^2=\varepsilon^2
\]
and
\[
(1-\varepsilon)^{-2}=1+2\varepsilon+O(\varepsilon^2).
\]
\end{proof}

\subsection{Proportional dilution}

Now add a new investor of mass $\lambda\in(0,1)$ who holds the market portfolio. Define the $(n+1)\times m$ matrix
\[
A^{(\lambda)}_{i\cdot}=(1-\lambda)A_{i\cdot}\quad (i\le n),
\qquad
A^{(\lambda)}_{n+1,j}=\lambda s_j.
\]

\begin{proposition}[Dilution by a market-weight investor]\label{prop:dilution}
The ownership matrix $A^{(\lambda)}$ satisfies
\begin{align*}
H_I(A^{(\lambda)}) &= (1-\lambda)^2 H_I(A)+\lambda^2,\\
H_S(A^{(\lambda)}) &= H_S(A),\\
M(A^{(\lambda)})   &= (1-\lambda)^2 M(A)+\lambda^2 H_S(A),\\
\Xdep(A^{(\lambda)}) &= (1-\lambda)\,\Xdep(A).
\end{align*}
\end{proposition}

\begin{proof}
The new investor marginals are $(1-\lambda)p_1,\dots,(1-\lambda)p_n,\lambda$, while the stock marginals remain $s$. This yields the formulas for $H_I$ and $H_S$. The identity for $M$ follows by squaring the original rows and the new market-weight row. For $\Xdep$, the first $n$ rows are scaled by $(1-\lambda)$ relative to the new product benchmark, while the new row coincides exactly with that benchmark:
\[
\Xdep(A^{(\lambda)})
=
\sum_{i=1}^n\sum_{j=1}^m
\frac{\bigl((1-\lambda)A_{ij}-(1-\lambda)p_i s_j\bigr)^2}{(1-\lambda)p_i s_j}
=
(1-\lambda)\Xdep(A).
\]
\end{proof}

Proposition \ref{prop:dilution} shows why $\Xdep$ is the correct object for nonproportional overlap. Capital injected at exact market weights reduces benchmark-adjusted dependence linearly because it adds no new dependence of its own.

\section{Non-identification from the marginals}\label{sec:nonid}

The row and column marginals do not identify the ownership matrix.

\begin{proposition}[A fixed-marginal family with variable micro concentration]
For $t\in[0,1/2]$, define
\[
A(t)=
\begin{pmatrix}
t & \tfrac12-t\\[4pt]
\tfrac12-t & t
\end{pmatrix}.
\]
Then, for every $t$,
\[
p=s=(\tfrac12,\tfrac12),
\qquad
H_I=H_S=\tfrac12,
\]
but
\[
M(A(t))=\tfrac14+4\left(t-\tfrac14\right)^2,
\qquad
\Xdep(A(t))=16\left(t-\tfrac14\right)^2.
\]
\end{proposition}

\begin{proof}
Every row and column sum equals $1/2$, so the marginals, and hence $H_I$ and $H_S$, are constant. Direct calculation gives
\[
M(A(t))=2t^2+2\left(\tfrac12-t\right)^2
       =\tfrac14+4\left(t-\tfrac14\right)^2.
\]
Since $p_i s_j=1/4$ for all cells,
\[
\Xdep(A(t))
=
4\sum_{i,j}\left(A_{ij}(t)-\tfrac14\right)^2
=
16\left(t-\tfrac14\right)^2.
\]
\end{proof}

\begin{corollary}[No scalar reconstruction from $(H_I,H_S)$]
There is no universal function $F$ such that $M(A)=F(H_I,H_S)$ for all ownership matrices $A$. Likewise, there is no universal function $G$ such that $\Xdep(A)=G(H_I,H_S)$.
\end{corollary}

\begin{proof}
In the family $A(t)$, the pair $(H_I,H_S)$ is constant, while both $M(A(t))$ and $\Xdep(A(t))$ vary with $t$.
\end{proof}

\begin{remark}[Copula viewpoint]
The non-identification problem is a discrete version of the classical question of recovering a joint law from its marginals. In the continuous setting, Sklar's theorem separates marginals from a copula, while Fr\'echet-type bounds constrain feasible joint laws \citep{sklar1959,frechet1951}. In the present framework, the matrix $A$ is the joint law, and $\Xdep(A)$ measures its departure from the product copula.
\end{remark}

\section{A worked example}\label{sec:example}

Consider the normalized ownership matrix
\[
A=
\begin{pmatrix}
0.30 & 0.10\\[3pt]
0.05 & 0.25\\[3pt]
0.15 & 0.15
\end{pmatrix}.
\]
Its entries sum to $1$.

\paragraph{Marginals.}
The investor marginals are
\[
p=(0.40,\,0.30,\,0.30),
\]
and the stock marginals are
\[
s=(0.50,\,0.50).
\]
Hence
\[
H_I=0.40^2+0.30^2+0.30^2=0.34,
\qquad
H_S=0.50^2+0.50^2=0.50.
\]
The corresponding effective counts are $N_I\approx 2.94$ and $N_S=2$.

\paragraph{Raw micro concentration.}
\[
M(A)=0.30^2+0.10^2+0.05^2+0.25^2+0.15^2+0.15^2=0.21,
\]
so $N_M=1/0.21\approx 4.76$.

\paragraph{Proportional benchmark and dependence.}
The proportional benchmark is
\[
ps^\top=
\begin{pmatrix}
0.20 & 0.20\\
0.15 & 0.15\\
0.15 & 0.15
\end{pmatrix},
\qquad
M(ps^\top)=H_I H_S=0.17.
\]
Thus the observed matrix is more concentrated at the cell level than the proportional benchmark.

Using Proposition \ref{prop:profiledecomp},
\[
\Xdep(A)=\sum_{i,j}\frac{A_{ij}^2}{p_i s_j}-1
        =\frac{0.09}{0.20}+\frac{0.01}{0.20}+\frac{0.0025}{0.15}+\frac{0.0625}{0.15}+\frac{0.0225}{0.15}+\frac{0.0225}{0.15}-1
        \approx 0.2333.
\]

\paragraph{Investor-level interpretation.}
The three investor profiles are
\[
q_1=(0.75,0.25),
\qquad
q_2=\left(\frac16,\frac56\right),
\qquad
q_3=(0.50,0.50).
\]
Their benchmark deviations from the market portfolio $s=(0.50,0.50)$ are
\[
\chi^2(q_1\|s)=0.25,
\qquad
\chi^2(q_2\|s)=\frac49\approx 0.4444,
\qquad
\chi^2(q_3\|s)=0.
\]
Therefore,
\[
\Xdep(A)=0.40(0.25)+0.30\!\left(\frac49\right)+0.30(0)=0.2333.
\]
Investor $2$ contributes the most to dependence because it is the most tilted away from market weights.

\paragraph{Fixed-marginal benchmarking.}
For these marginals, the fixed-marginal minimum and maximum of $M$ are
\[
M_{\min}=0.17,
\qquad
M_{\max}=0.30.
\]
Hence
\[
\Psi(A\mid p,s)=\frac{0.21-0.17}{0.30-0.17}\approx 0.308.
\]
Thus the observed matrix lies roughly one-third of the way from the fixed-marginal minimum to the fixed-marginal maximum of raw sparsity.

\paragraph{Spectral analysis.}
The whitened matrix is
\[
K=D_p^{-1/2}AD_s^{-1/2}
=
\begin{pmatrix}
0.6708 & 0.2236\\
0.1291 & 0.6455\\
0.3873 & 0.3873
\end{pmatrix},
\]
whose singular values are
\[
\sigma_1(K)=1,
\qquad
\sigma_2(K)\approx 0.4830.
\]
Therefore,
\[
\rho(A)=\sigma_2(K)\approx 0.4830,
\qquad
\rho(A)^2\approx 0.2333=\Xdep(A),
\]
which is exactly what Theorem \ref{thm:spectral} predicts in a $3\times2$ system.

\paragraph{Summary.}
The matrix has moderate investor concentration ($H_I=0.34$), balanced stocks ($H_S=0.50$), moderately elevated raw cell concentration ($M=0.21$), and nontrivial but not extreme benchmark-adjusted dependence ($\Xdep\approx 0.23$, $\rho\approx 0.48$). The diagnostics point to a single main source of specialization: investor $2$'s strong tilt toward stock $2$.

\section{Empirical implementation and research agenda}\label{sec:empirical}

The paper suggests that empirical ownership analysis should not rely on a single concentration number. A practical dashboard is
\[
(H_I,\ H_S,\ M,\ \Psi,\ \Xdep,\ \rho).
\]
Each object answers a different question:
\begin{align*}
H_I &: \text{How concentrated is aggregate wealth across investors?}\\
H_S &: \text{How concentrated is aggregate capitalization across stocks?}\\
M   &: \text{How concentrated are occupied investor--stock cells?}\\
\Psi &: \text{How sparse is the matrix relative to the sharp fixed-marginal range?}\\
\Xdep &: \text{How far is the matrix from proportional ownership?}\\
      &\text{How much total active-variance capacity does that create?}\\
\rho &: \text{What is the dominant nontrivial overlap mode?}\\
     &\text{How strong is worst-case idiosyncratic transmission?}
\end{align*}

A practical empirical workflow is straightforward. Normalize the holdings matrix to obtain $A$. Compute the marginals $(p,s)$ and the Herfindahl indices $(H_I,H_S)$. Then compute $M(A)$, $\Xdep(A)$, and the leading nontrivial singular value $\rho(A)$ from the whitened matrix. If fixed-marginal benchmarking is substantively relevant, solve the convex program $\min_{B\in\T(p,s)} M(B)$ and the companion maximization problem to obtain $\Psi(A\mid p,s)$.

The dashboard should be read in layers. A large value of $M(A)$ may reflect concentrated investors, concentrated stocks, or sparse assignment. A large value of $\Xdep(A)$ can only come from nonproportional ownership profiles. A large value of $\Psi(A\mid p,s)$ means the matrix is concentrated relative to what the marginals permit, even if $\Xdep(A)$ is modest. The three measures therefore diagnose distinct failure modes of diversification.

Through Theorems \ref{thm:firesale} and \ref{thm:alpha}, the same dashboard also acquires dynamic content. The quantity $\rho(A)$ is the sharp coefficient mapping weighted idiosyncratic liquidation shocks into weighted cross-sectional price pressure in the linearized fire-sale system. The quantity $\Xdep(A)$ is the isotropic average-case capacity of the ownership system to transform centered stock-return dispersion into benchmark-relative investor dispersion. In that sense, the paper's static overlap measures become dynamic transmission coefficients once a linear mechanism is specified.

The framework also fits naturally with the systemic-risk literature. Concentrated ownership can make prices fragile when large holders face correlated liquidity shocks \citep{greenwoodthesmar2011}. Overlapping portfolios create indirect network linkages that amplify fire-sale dynamics \citep{greenwoodlandierthesmar2015,poledna2021}. Statistically validated portfolio-overlap networks extend the same idea to network inference \citep{gualdi2016}. The point of the present theory is not to replace those models, but to provide a measurement layer that isolates scale concentration, feasible sparsity, nonproportional overlap, and the leading linear transmission coefficients before one imposes a richer contagion mechanism.

\begin{remark}[Computational cost]
For an $n\times m$ holdings matrix, computing $(p,s)$, $(H_I,H_S)$, $M$, and $\Xdep$ is $O(nm)$. The singular value decomposition of the whitened matrix costs $O(nm\min\{n,m\})$, but in large sparse panels the dominant quantity $\rho(A)$ is obtainable by truncated methods. The fixed-marginal minimization of $M$ is a strictly convex quadratic program with $nm$ variables and $n+m$ equality constraints; the maximization can be approached through extreme-point search or specialized transportation-polytope algorithms.
\end{remark}

\section{Conclusion}

Ownership concentration is a layered property of a bipartite system. The first layer is investor concentration, measured by the row-marginal Herfindahl index. The second is stock concentration, measured by the column-marginal Herfindahl index. The third is benchmark-adjusted dependence, measured by the deviation of the joint ownership matrix from the proportional benchmark. Raw micro concentration belongs in the story, but it is not a substitute for any one of those layers.

The mathematical structure is sharp. Raw micro concentration admits exact row and column decompositions. Benchmark-adjusted dependence admits exact profile decompositions and a multiscale aggregation law. For fixed marginals, minimum raw concentration is unique, whereas maximum raw concentration can be attained at a sparse extreme point whose support graph is a forest. Spectrally, the entire nontrivial dependence structure is encoded in the singular values of the whitened matrix after the trivial rank-one mode is removed.

The framework also bridges static topology and dynamic transmission. In the linearized fire-sale system, the dominant nontrivial singular value $\rho(A)$ is the exact worst-case coefficient for idiosyncratic shock propagation after size normalization. In the benchmark-relative return system, $\rho(A)^2$ is the worst-case active-variance multiplier and $\Xdep(A)$ is the isotropic average-case active-variance capacity. The same residual operator $L$ therefore governs both liquidation spillovers and cross-sectional alpha dispersion. That duality is the paper's strongest conceptual point: overlap is not merely a descriptive pattern in the holdings matrix; it is a transmission mechanism.

The conceptual payoff is that several quantities that look superficially similar answer different questions. $M(A)$ measures cell concentration. $\Psi(A\mid p,s)$ measures concentration relative to the sharp feasible range implied by the marginals. $\Xdep(A)$ measures departure from proportional ownership and total quadratic active-variance capacity. $\rho(A)$ identifies the strongest nontrivial overlap mode and the sharp idiosyncratic transmission coefficient. Treating these as interchangeable obscures both the geometry and the economics.

\paragraph{Signed ownership and leverage.}
A natural next step is to admit short positions. Write the raw exposure matrix as $Y=Y^+-Y^-$ with $Y^\pm\ge 0$ and $Y^+_{ij}Y^-_{ij}=0$, and normalize by gross exposure:
\[
G=\sum_{i,j}(Y^+_{ij}+Y^-_{ij})>0,
\qquad
A^\pm=\frac{Y^\pm}{G},
\qquad
A=A^+-A^-.
\]
Then $\norm{A}_1=1$. Define the gross marginals
\[
p_i^{\mathrm g}=\sum_j (A_{ij}^+ + A_{ij}^-),
\qquad
s_j^{\mathrm g}=\sum_i (A_{ij}^+ + A_{ij}^-),
\]
and the net marginals
\[
p_i^{\mathrm n}=\sum_j A_{ij},
\qquad
s_j^{\mathrm n}=\sum_i A_{ij}.
\]
The quadratic micro concentration $\norm{A}_F^2$ remains well defined, but the raw signed fixed-net-marginal problem is unbounded unless gross exposure is fixed or leverage is capped. After gross normalization one obtains a bounded mixed-sign feasible set, but its geometry is no longer the standard transportation polytope. When the aggregate net exposure
\[
\eta=\sum_i p_i^{\mathrm n}=\sum_j s_j^{\mathrm n}
\]
is nonzero, the canonical rank-one net benchmark is
\[
B=\eta^{-1}p^{\mathrm n}(s^{\mathrm n})^\top,
\]
which shares the net marginals of $A$. A gross-whitened signed dependence measure is then
\[
\Xdep_{\mathrm{signed}}(A)
=
\sum_{i,j}\frac{(A_{ij}-B_{ij})^2}{p_i^{\mathrm g}s_j^{\mathrm g}}
=
\bigl\|(D_{p^{\mathrm g}})^{-1/2}(A-B)(D_{s^{\mathrm g}})^{-1/2}\bigr\|_F^2.
\]
Market-neutral systems with $\eta=0$ require a genuinely two-sided long/short benchmark rather than a single rank-one net benchmark. Developing the corresponding mixed-sign transport geometry is a substantial extension, but the gross-whitened formulation above provides the right quadratic starting point.

The framework is deliberately first-order. It does not model endogenous liquidation choice, nonlinear price impact, control rights, derivatives, or equilibrium feedback. Its role is more basic: given a holdings matrix, it identifies what part of observed concentration comes from investor size, stock size, feasible sparsity, genuine dependence in ownership profiles, and the leading linear transmission channels. That decomposition is a prerequisite for a more granular empirical science of crowding, fragility, and ownership-based systemic risk.

\appendix

\section{Rényi sensitivity analysis}\label{app:renyi}

The main text focuses on the quadratic case because it uniquely supports the collision-probability, Frobenius-norm, and $\chi^2$-based spectral representations above. The broader Rényi family is nevertheless useful as a robustness device.

\begin{definition}[Rényi concentrations]
For a probability vector $w=(w_1,\dots,w_d)$ and $\alpha>0$, $\alpha\neq 1$, define
\[
C_\alpha(w)=\sum_{k=1}^d w_k^\alpha,
\qquad
N^{(\alpha)}(w)=C_\alpha(w)^{1/(1-\alpha)}.
\]
For the ownership matrix,
\[
H_I^{(\alpha)}=C_\alpha(p),
\qquad
H_S^{(\alpha)}=C_\alpha(s),
\qquad
M^{(\alpha)}(A)=\sum_{i,j}A_{ij}^\alpha.
\]
\end{definition}

\begin{proposition}
Under proportional ownership $A=ps^\top$,
\[
M^{(\alpha)}(A)=H_I^{(\alpha)}H_S^{(\alpha)}
\qquad\text{for every }\alpha>0,\ \alpha\neq 1.
\]
Equivalently,
\[
N_M^{(\alpha)}=N_I^{(\alpha)}N_S^{(\alpha)}.
\]
\end{proposition}

\begin{proof}
If $A_{ij}=p_i s_j$, then
\[
M^{(\alpha)}(A)=\sum_{i,j}(p_i s_j)^\alpha
=\left(\sum_i p_i^\alpha\right)\left(\sum_j s_j^\alpha\right)
=H_I^{(\alpha)}H_S^{(\alpha)}.
\]
\end{proof}

\begin{remark}[Why the quadratic case remains central]
The order $\alpha=2$ is distinguished for four reasons. First, it yields the Frobenius representation $M(A)=\norm{A}_F^2$, which connects directly to the singular values. Second, it is exactly the collision probability of two independent draws. Third, the benchmark-adjusted dependence measure $\Xdep(A)$ is a Pearson $\chi^2$ object, the natural normalized quadratic companion to $M(A)$. Fourth, the dynamic transmission theorems of Section \ref{sec:dynamic} are genuinely quadratic: fire-sale severity and aggregate active variance are both squared-norm objects generated by the same residual operator. The Rényi family is therefore best viewed as a sensitivity layer around this fundamentally quadratic core.
\end{remark}

\bibliographystyle{plainnat}
\bibliography{ownership_concentration_unified}

\end{document}